\documentclass{article} 
\usepackage[margin=1in]{geometry}
\usepackage[nolists,nomarkers,tablesfirst]{endfloat} 
\usepackage{amsmath} 

\usepackage{amsmath,amssymb,amsfonts}
\usepackage{authblk}

 \title{The restrictive conditions to solve LTI Systems by Ordinary Differential Equations \\ 
 \Large And how State Space Representations help to expand them }

\author[1]{Alexandre Sanfelici Bazanella$^1$ \and Trist\~ao Garcia}

\affil[1]{Department of Automation and Energy \\ Universidade Federal do Rio Grande do Sul \\ Porto Alegre - RS - Brazil}

\date{\empty}
\usepackage{xcolor}
\newcommand{\bydef}{\stackrel{\Delta}{=}}
\newcommand{\beq}{\begin{equation}}
\newcommand{\eeq}{\end{equation}}

\providecommand{\com[2]}{\begin{tt}[#1: #2]\end{tt}}
\definecolor{red}{rgb}{1,0,0}

\definecolor{blu}{rgb}{0,0,1}

\definecolor{gre}{rgb}{0,0.7,0.3}

\newcommand{\vecx}{\mathbf{x}}
\newcommand{\cqfd}{\hfill \rule{2mm}{2mm}\medbreak\indent}

\newtheorem{definition}{Definition}
\newtheorem{theorem}{Theorem}
\newtheorem{problem}{Problem}
\newtheorem{example}{Example}

\begin{document}

\maketitle

	\pagestyle{empty}
	\thispagestyle{empty}

%
%

\section{Introduction}\label{sec:intro}

Ordinary linear differential equations 
(aka linear ODE's) 
occupy a special place in the theory of physics and engineering.
Their convenience and importance in the formulation and in the solution of
a wide variety of problems in circuits, mechanics, control, and all sorts
of applications can hardly be overstated. 
The study of ODEs is the entrance door to systems theory, which 
occupies a central role in Electrical Engineering in general and in Control Systems in particular. 
Accordingly, they receive special attention in the curricula of Electrical Engineering
and in Control Engineering. 

An ODE can be seen as an operator on the space of functions, and it has a nontrivial kernel,
so any ODE admits an infinite number of solutions -  specifically, a subspace of dimensions $n$,
where $n$ is the order of the ODE. One can specify the
value of the solution and of its derivatives of order up to $n-1$ at a given point in time - say $t=0$ - to
obtain a unique solution. These are the 
initial value problems (IVPs), which are the ones of utmost  interest in our field. 
Indeed, many course hours are usually allocated to their study in EE programs around the globe. 
In our University, in the engineering programs in Electrical Engineering (100 new students per year) and
Control and Automation Engineering (33),  there is a one semester course dedicated exclusively 
to the study of differential equations, plus a one semester course on Signals and Systems 
and two semesters on Circuit Theory, where ODEs occupy a significant part of the semester. 

There is ample choice of literature, with several celebrated textbooks, for instance, 
\cite{Kaplan_ingles, Boyce_ingles, Haykin_ingles, Lathi_ingles, Desoer_ingles, IMPA, Geromel, Orsini}
just to name a few among the ones that the authors 
have seen most often in their (quite long) teaching experience. This literature covers an ample variety of topics and 
applications, and dedicates a lot of space to the study of IVPs, but it is laconic at best about the initial conditions
themselves. To illustrate what we mean, take the simple example of a second-order ODE:
\begin{equation}\label{example_obs}
\ddot{y}(t) + 5\dot{y}(t) + 6 y(t) = \dot{u}(t) + u(t)
\end{equation}
for which we want to determine the step response with some nonzero initial condition. In other words, we are given 
$u(t) =0~ \forall t<0$, $u(t) = 1~\forall t>0$ and  initial values of $y(t)$ and its derivative. 

The most ``popular" way to find the solution for such an IVP is to use the Laplace Transform.
Applying the Laplace Transform to the ODE (\ref{example_obs}) transforms it into an algebraic
equation in which $Y(s)$ (the Laplace Transform of $y(t)$) can be isolated, yielding:
\begin{equation}\label{solucao_ingenua}
Y(s) = \frac{s+1}{s^2 +5s+6 } U(s) +  \frac{\dot{y}(0) + (s+5) y(0) + u(0)}{s^2 +5s+6}
\end{equation}

Once we have $Y(s)$ we can obtain the solution $y(t) $ by calculating the inverse Transform of $Y(s)$, which 
in almost all cases consists just in performing a partial fraction expansion of $Y(s)$ and looking at a table. 
This is really quite standard and (arguably) very well known. But actually $u(t)$ is discontinuous at $t=0$, that is, 
$u(0)$ does not really exist. So, what should we do? Should we use $u(0) = 0$ in this formula, 
or perhaps  $u(0) = 1$? What about $y(0)$ and $\dot{y}(0)$ that also 
appear there: are these quantities well defined or is $y(t)$ also discontinuous at $t=0$, like the input $u(t)$?
If they are discontinuous (spoiler: $\dot{y}(t)$ is), what value should we use as initial condition in the solution?

These are very basic questions, and very fundamental ones, and yet they are not treated uniformly in the literature.
In the calculus books, the right-hand side (RHS) of the ODE is just assumed (mostly tacitly) to be continuous, which also implies
that the solution $y(t)$ is continuous, thus ignoring the cases where discontinuities arise.
In most of the literature devoted to signals and systems and to circuit theory this class of problems,
which are predominant in engineering applications, is acknowledged. Still, there are some textbooks in which this
point is not mentioned at all, leaving the student at their own devices to decide what to do.
To cover these cases, some texts will present an analysis that will result in using $u(0) =0$
(as in \cite{Lathi_ingles}, among others) whereas some others would have us use $u(0)=1$ instead
(as in \cite{Haykin_ingles}, for example), not always with a convincing justification for
this choice. This lack of uniformity in the literature brings some discomfort and, more importantly, the assumptions
under which each choice is valid and/or convenient are seldom discussed (if at all).

In this note we revisit the solution of initial value problems with the aim to provide a solution 
that is clearer, with firm theoretical foundations and boundaries. 
In our path to deriving this theory we'll confront issues of existence of solutions and 
the relationship of ODEs with state space representations (SSRs). Such relationships between the 
different representations is quite relevant, and yet little unexplored in the literature. 
The fact that one is looking for a solution in the realm  of real numbers for an equation whose RHS does not belong to it (since $\dot{u}(0)$ is not finite) may  also be a source of anxiety that we prefer to remove whenever possible, and the SSR allows to do just that.

\section{The limits of traditional theory - the continuous case}

Let us start our technical treatment with a formal definition of the problem itself. 
A general linear ODE is defined as
\begin{equation}\label{EDOLcc}
\stackrel{(n)}{y} (t) + \sum_{i=1}^{n} a_i \stackrel{(n-i)}{y} (t)= \sum_{j=0}^{n} b_j \stackrel{(n-j)}{u} (t)
\end{equation}
where $a_i, ~b_j\in\mathbb{R}$ are constant coefficients, the independent variable is the input signal $u(t )$
and $y(t)$ is the output to be determined by solving the equation. If the first $r$ coefficients $b_j$ are zero and $b_{r}\neq 0$, 
one says that $r$ is the {\em relative degree} of the ODE and define $m=n-r$.

Engineering applications involving such ODEs, in particular in electrical engineering,
are mostly about solving  initial value problems (IVPs).
An IVP consists of finding the unique solution of (\ref{EDOLcc}) for a given set of {\em initial conditions} 
$y(0)$, $\dot{y}(0)$, $\ldots $, $\stackrel{(n-1)}{y}(0)$. For convenience of notation we define the following
vectors: 
$$
{\cal Y}(t) = \left[\begin{array}{c}
\stackrel{(n-1)}{y}(t) \\ \stackrel{(n-2)}{y}(t) \\ \vdots \\ y(t)  \end{array}\right] 
\quad
{\cal U}(t) = \left[\begin{array}{c}
 \stackrel{(n-1)}{u}(t)   \\  \stackrel{(n-2)}{u}(t)   \\ \vdots \\  u(t) \end{array}\right] .
$$ 
The IVP can then be spelled as: given an ODE in the form (\ref{EDOLcc}), a
signal $u(t)~\forall t > 0$ and a value ${\cal Y}(0)$, find the solution
$y(t)~\forall t>0$. 

Whatever we'll discuss in this paper concerns the solutions of IVPs themselves, regardless of the method 
used for their solution. Still, it is easier to communicate using a closed expression fo the solution, so we'll adopt it. 
The prescribed solution for IVPs is the application of the Laplace Transform to the ODE. 
The Laplace Transform is defined as
\beq\label{Laplace}
{\cal L}\{ x(t) \} = \int_0^{\infty} x(t) e^{-st} dt
\eeq
and it possesses the following property
\beq\label{propriedade}
{\cal L}\{ \dot{x}(t) \} = s{\cal L}\{ x(t) \} -x(0) .
\eeq
It is the property (\ref{propriedade}) that justifies this prescription, for it allows to turn 
the differential equation (\ref{EDOLcc}) into an algebraic equation. Indeed, applying the
Laplace Transform to each term of (\ref{EDOLcc}) and isolating the Transform of the 
unknown $y(t)$ gives the following closed-form expression for the solution, which can be found
in textbooks like \cite{Haykin_ingles}.

\begin{eqnarray}
Y(s) & = &\frac{1}{s^n +\sum_{i=1}^{n} a_i s^{n-i} } \sum_{j=0}^{n} b_j s^{n-j} U(s) - \nonumber \\
&& \frac{1}{s^n +\sum_{i=1}^{n} a_i s^{n-i} } \sum_{j=0}^{n-1} b_j \sum_{k=1}^{n-j} s^{n-k-j} \stackrel{(k-1)}{u} (0) + \nonumber \\
&& \frac{1}{s^n +\sum_{i=1}^{n} a_i s^{n-i} } \sum_{i=0}^{n-1} a_i \sum_{k=1}^{n-i} s^{n-k-i} \stackrel{(k-1)}{y} (0) \nonumber \\
&& \label{ingenua}
\end{eqnarray}
where $Y(s)\bydef {\cal L}\{y(t)\}$ and $a_0=1$. 

Alas, as illustrated by the simple example above, this problem formulation and this solution are
not precise enough to be applied in all cases of interest, for in many of those
the quantities in this formula are not defined - namely, when ${\cal U} (t)$ and/or ${\cal Y}(t)$ are 
discontinuous at $t=0$.
The use of Laplace Transforms for the solution of differential and integral equations is 
as classic a subject as can be. 
Pierre-Simon Laplace himself and other mathematicians of his time - around the year 1800, that is - already
studied such integral transformation with the aim to solve differential equations, 
and the contemporary formulation presented above dates back to almost a century ago \cite{Doetsch}. 
Its application to the analysis of electrical circuits,
in which discontinuities appear very often, dates back to at least as
early as the 1930's and plenty of work has been done in the following
decades. In \cite{Ghizzetti1938}  the author dedicates his long text to the study of ``the difficulties
in defining the initial conditions'' in electrical circuits and proposes a consistent way to do that.
This issue is not restricted to circuits either - see, for instance, \cite{Vibet1988}.
But somehow this concern and these results have not found their way into the didactic literature
available today, which is laconic about it. 


Still, if $u(t)$ and $y(t)$ are smooth enough at $t=0$ these 
quantities are well defined and in this case the issue is eliminated. So we start with a formal statement of
the conditions under which this happens, whose proof is delayed to the next Section. 

\begin{theorem}\label{teoreminha}
If and only if $u(t)$ and its first $m-1$ derivatives are continuous at $t=0$ then  ${\cal Y}(t)$ is continuous at $t=0$.
\cqfd
\end{theorem}

We have thus established the conditions for validity of the traditional solution
(\ref{ingenua}): it is valid for any initial condition and for all inputs satisfying the conditions
of Theorem \ref{teoreminha}.
Under these conditions, all required signals are continuous at $t=0$, thus the solutions
prescribed in the didactic literature are well defined. However, these conditions
are very restrictive. For instance,  the most common input signal is perhaps the Heaviside step function $d(t)$, which is defined as 
$$
d(t) = 1~\forall t>0 \quad  d(t) = 0~\forall t<0 
$$
for which the solution (\ref{ingenua}) is well defined only if $m=0$. 

The question of what to do when these conditions are not satisfied
is not settled in the literature. In order to answer to it, we will first discuss the equivalence between
ODEs and SSRs.

\section{State space representation}

When the condition of Theorem \ref{teoreminha} is not satisfied, the RHS of equation (\ref{EDOLcc})
goes to infinity, as at least one of the terms is a ``singular function'' - a Dirac's $\delta$ function
or one of its derivatives. A formally-oriented mind may ask what does it mean a solution to an
equation whose RHS does not even exist within the realm of real numbers. To answer to this
question one can appeal to the treatment of singular functions 
or to the formalism of distributions, 
as done in some textbooks. 
In this paper we provide an alternative: we give the pragmatic
answer that a solution of the ODE is by definition the solution of
its equivalent dynamic equation, to be defined in the sequel. 
We start with an example. 

\begin{example}\label{example}
Consider the ODE given in the following equation 
\begin{equation}\label{example_not_obs}
\ddot{y}(t) + 6\dot{y}(t) + 5 y(t) = \dot{u}(t) + u(t) .
\end{equation}

To avoid for the moment discontinuity issues, consider its solution with the initial
condition ${\cal Y}(0) = [ 1 ~ 0 ]^T$ and a ramp input
$u(t) = t\quad \forall t \geq 0$, which satisfies the conditions of
Theorem \ref{teoreminha}. In this case the formula (\ref{ingenua}) can be applied
and results in 
\begin{equation}\label{solex2}
y(t)= \frac{1}{5} t - \frac{1}{25} (1-e^{-5t}) + \frac{1}{4} (e^{-t} - e^{-5t} )
\end{equation}

The input-output relationship defined by this ODE can be conveniently described
by the operator in Laplace domain, the transfer function
$$G(s)= \frac{Y(s)}{U(s)} = \frac{s+1}{s^2+6s+5} .$$
The problem of obtaining SSRs with a given transfer function 
is standard material in linear system theory, called the realization problem. 
The realization of $G(s)$ in observable canonical form is:
\begin{eqnarray}
\dot{\vecx} (t) & = & 
\left[\begin{array}{cc} 0 & - 5 \\ 1 & -6 \end{array}\right] \vecx (t) + \left[\begin{array}{c} 1  \\ 1 \end{array}\right] u(t) \label{SSex} \\
y (t) & = & \left[\begin{array}{cc} 0 & 1  \end{array}\right] \vecx (t) \label{SSeo}
.
\end{eqnarray}
But realization theory occupies itself only with the input-output relationship. The
question of whether or not the dynamic equation (\ref{SSex})(\ref{SSeo}) 
has the same response to initial conditions as the ODE (\ref{example_not_obs}) from which it 
originated is not answered.

The solution of the SSR (\ref{SSex})(\ref{SSeo})  to the same input and the same initial condition can be obtained by
the standard formula 
\begin{eqnarray}\label{SSsolution}
\vecx (t) & = &  e^{At} \vecx (0) + \int_0^t e^{A(t-\tau)}B u(\tau ) d\tau \\
 y(t) & = & C \vecx (t) + D u(t)
\end{eqnarray}
with the initial condition $\vecx (0)$ corresponding to the given ${\cal Y} (0)$, which 
can be calculated as follows. Taking the derivative of the output in (\ref{SSex})(\ref{SSeo})
\begin{equation}\label{qwe}
\dot{y} (t)= \left[\begin{array}{cc} 0 & 1  \end{array}\right] \dot{\vecx} (t) = 
\left[\begin{array}{cc} 1 & -5  \end{array}\right] \vecx (t) + u(t)
\end{equation}
and applying this equation at $t=0$ we have a system of two equations
with two unknowns that can be solved for $\vecx (0)$:
$$
\left[\begin{array}{c} 1 \\ 0 \end{array}\right] = \left[\begin{array}{cc}   1 & -5 \\ 0 & 1 \end{array}\right] \vecx (0)
$$
whose solution is $\vecx (0) = [ 1 ~ 0 ]^T$. Substitution of this value and the input into  
(\ref{SSsolution}) yields exactly the same expression as (\ref{solex2}). So, the answer is yes, 
the complete response of the dynamic equation equals the complete response 
of the ODE, and it is easy to verify that this happens for any input and any IC. 
The two descriptions are equivalent, as they
describe the exact same behaviors. \cqfd
\end{example}

\subsection{The equivalence}

Motivated by our example, and in the classical spirit of the theory in \cite{ZadehDesoer}, we define the following.
%

\begin{definition}\label{def_equivalence}
A state space representation is said to be equivalent to a differential
equation in the form (\ref{EDOLcc}) if for any input signal satisfying the conditions
of Theorem \ref{teoreminha} and any initial condition
${\cal Y}(0)$ there exists some initial condition $\vecx (0)$ such that 
the solution $y(t)$ of the ODE is the same as the solution of the SSR. \cqfd
\end{definition}

A dynamic equation, or state space representation (SSR), is  in the general form 
\begin{eqnarray}
\dot{\mathbf{x}} (t)& = &  A \mathbf{x} (t)+ Bu (t) \label{eq_estado} \\
y (t) & = & C \mathbf{x}(t) + Du(t) . \label{eq_saida} 
\end{eqnarray}
    
The equivalence between an ODE (\ref{EDOLcc}) and a SSR in the form
(\ref{eq_estado})-(\ref{eq_saida})
can be analyzed calculating the derivatives of the output $y(t)$ from (\ref{eq_saida}):
\begin{eqnarray*}
\dot{y} (t) & = & C \dot{\mathbf{x}}(t) + D\dot{u}(t)  = CA \mathbf{x}(t) +CBu(t) + D\dot{u}(t) \\
\ddot{y}(t) & = & CA \dot{\mathbf{x}}(t) +CB\dot{u}(t) + D\ddot{u}(t) =  CA^2 \mathbf{x}(t) + CAB u(t)+ CB\dot{u}(t) + D\ddot{u}(t) 
\nonumber \\
\vdots & & \vdots \nonumber \\
\stackrel{(n-1)}{y}(t) & = & CA^{n-1} \mathbf{x}(t) + CA^{n-2} B u(t) + CA^{n-3} B \dot{u}(t) + \ldots + D \stackrel{(n-1)}{u}(t) 
\end{eqnarray*}
where one recognizes the Markov parameters $h_i$ \cite{Chen}:
\begin{equation}\label{Markovpar}
h_0 = D \quad\quad h_i = CA^{i-1}B~~ i= 1, 2, \ldots
\end{equation}

This set of equations can be written as one single matrix equation:
\begin{equation}\label{markov}
\left[\begin{array}{c}
\stackrel{(n-1)}{y}(t)  \\ \stackrel{(n-2)}{y}(t) \\ \vdots  \\ \dot{y}(t) \\ y(t)  \\ 
\end{array}\right]
= 
\left[\begin{array}{c}
CA^{n-1} \\ CA^{n-2} \\ \vdots \\ CA  \\ C \\ 
\end{array}\right]
\vecx
+
\left[\begin{array}{ccccc}
h_0 & h_1 & h_2 & \ldots & h_{n-1}  \\
 0 & h_0 & h_1 & \ldots & h_{n-2}  \\ 
 0 & 0 & h_0  & \ldots & h_{n-3}  \\
\vdots &&& \vdots \\
0 & 0 & 0 & \ldots & h_0
\end{array}\right]
\left[\begin{array}{c}
\stackrel{(n-1)}{u}(t)   \\ \stackrel{(n-2)}{u}(t)   \\ \vdots \\ \dot{u}(t) \\ u(t)  \\ 
\end{array}\right]
\end{equation}
or, in more compact form,
\beq\label{Markov}
{\cal Y}(t) = O \vecx (t) + M {\cal U}(t)
\eeq
where we have defined the observability matrix  $O$ and the Markov parameter matrix $M$.

To be equivalent, both the ODE and the SSR must have the same response to any input and to 
any initial condition.
The problem of equivalence from an input-output point if view only is fully 
solved by the realization theory. The input response is determined by the transfer
function $G(s)$,  which can be calculated from the ODE as
$$
G(s) =  \frac{\sum_{j=0}^{n} b_j s^{n-j}}{s^n +\sum_{i=1}^{n} a_i s^{n-i} } 
$$
and from the SSR as
$$
G(s) = C(sI-A)^{-1} B + D .
$$ 
Clearly, the input responses of the ODE and of the SSR are the same for any input if the transfer functions are the same. 
As for the response to initial conditions, it will suffice to find an initial condition for
the SSR - a vector $\mathbf{x}(0)$, that is - which corresponds to the given initial
condition for the ODE - the vector ${\cal Y}(0)$. This is easily 
accomplished by writing  equation (\ref{Markov}) for $t=0$ then solving it for  $\mathbf{x}(0)$. 
It is only important to realize that 
finding such a solution relies on the observability matrix $O$ being invertible, in other words, the
dynamic equation being observable. 

In conclusion, the differential equation (\ref{EDOLcc}) and the state space representation (\ref{eq_estado})-(\ref{eq_saida})
are equivalent if the following three conditions are satisfied:
\begin{enumerate}
\item they have the same order $n$
\item $C (sI-A)^{-1}B +D =  \frac{\sum_{j=0}^{n} b_j s^{n-j}}{s^n +\sum_{i=1}^{n} a_i s^{n-i} }$ 
\item the state space representation is observable.
\end{enumerate}

The theory of realizations - see \cite{Chen}, for example - occupies itself only with transfer functions, that is,
it treats the problem of finding an SSR that has a given input-output relationship. 
The equivalence concept expressed in Definition \ref{def_equivalence} is different: 
we are looking for an SSR that possesses not only the same input response
as the ODE, but also the same response to initial conditions. 

Equation (\ref{markov}) also allows to prove Theorem \ref{teoreminha},
because it provides an algebraic relation between the derivatives of $y(t)$ and those of $u(t)$.
Since the state is continuous, a particular derivative of $y(t)$ is discontinuous if and only if
at least one of the terms on $u(t)$ appearing in its expression is discontinuous. It is easy to verify in (\ref{markov}) that the
expression of  $\stackrel{(j)}{y}(t)$ for any given $j$ depends only on ${\cal U}_{j-r}(t)$ and not on
higher derivatives of $u(t)$, which immediately implies the statement of Theorem \ref{teoreminha}.
%
%
%

%
%



An additional concern is that when the conditions of Theorem \ref{teoreminha}
are not satisfied the RHS of the ODE does not even exist, so one might ask what is - conceptually 
speaking - a solution to this equation. 
Dynamic equations have a key property that allows us to escape from this tight spot:
their solutions $\vecx (t)$ are continuous
for any {\em finite} $u(t)$. A dynamic equation ``does not care'' whether or
not the input $u(t)$ satisfies the conditions of Theorem \ref{teoreminha}, 
as long as it is finite-valued. So, for such input signals we can
{\em define} the solution of the ODE as the solution of the equivalent
SSR. In so doing we avoid resorting to infinite-valued
``functions'' and other psychedelic mathematical objects.

State space representations have been around for quite a while, in the systems theory in general and
in circuit theory as well.  
It was in 1957 \cite{AMatrix} that it was first proposed to use this representation in 
linear ``network analysis'' (the name used for ``circuit theory" in those days),
which was called ``the A-matrix description''. A textbook on state space representations written 
by hardcore electrical engineers was available
as early as 1963 \cite{ZadehDesoer}.
Yet, not many textbooks on circuit theory 
make ample use of SSR's, even though they scale much better than ODE's for large
circuits, provide convenient proofs for many properties of linear circuits and systems in general, and are already
covered by other disciplines in an Electrical Engineering program. Textbooks on signals and systems
tend to include state space representations, but with little to no connection to the study of differential equations.

\section{The discontinuous case}

In the classical theory of differential equations, as presented in the textbooks on calculus, it is usually
assumed that the RHS of (\ref{EDOLcc}) is continuous, or at least finite. This, as seen in Theorem \ref{teoreminha},
easily solves the question of initial conditions, but this assumption is not usually made explicit in the literature.
Perhaps more importantly, this solution is as simple as it
is unrealistic, for in an enormous amount of  applications it is not satisfied, implying that 
${\cal U}(t)$ and/or ${\cal Y}(t)$ are discontinuous at $t=0$. 
To deal with these cases, it is convenient to define the following notation and nomenclature.
 
\begin{definition}
Let  $x(\cdot ):~ \mathbb{R} \rightarrow \mathbb{R}$. We define 
$$
x(0^-) \bydef \lim_{t \rightarrow 0^-} x(t)
$$
and 
$$
x(0^+)  \bydef \lim_{t \rightarrow 0^+} x(t) .
$$
Since these are, in general,  two different quantities, we also give them different names:
$x(0^+)$ will be named the  {\em first condition} of $x(t)$, whereas 
$x(0^-)$ will be named its  {\em previous condition}. 
When $x(t)$ is continuous at $t=0$, we will use the
simpler notation $x(0)$ and will keep the nomenclature
{\em initial} condition.
\cqfd
\end{definition}

This notation is commonplace in the literature. But here we have also chosen to define the new
nomenclature {\em previous} and {\em first} conditions, as opposed to {\em initial} conditions, to make it clearer
that these are different things and they can not always be interchanged for each other. 


\subsection{Reformulating the IVP}

Since all problems come from discontinuities at $t=0$, the most natural solution to these
problems is to forget about this time instant and focus only on $t>0$. We do just that and
redefine the IVP as follows.

\begin{problem}\label{problema_certo}

Given:
\begin{enumerate}
\item an input signal $u(\cdot ): \mathbb{R}^+ \rightarrow \mathbb{R}$ that is finite 
for all $t \in (0, t_f )$,
\item a vector of first conditions ${\cal Y}(0^+)$,
\end{enumerate}
find $y(t)~\forall t\in (0, t_f )$ satisfying the ODE (\ref{EDOLcc}) for this input and this first condition. 
\cqfd
\end{problem}

The problem thus posed is formally perfect, in the sense that there are neither potentially 
undesired quantities in its statement - infinite valued functions - nor undefined quantities in its
 solution - the values at $t=0$. And it can always be solved in the prescribed way by just replacing
the values of ${\cal U}(0)$ and ${\cal Y}(0)$ in equations (\ref{ingenua}) by their values at $t=0^+$. 

This formulation has the merit that it keeps all technical difficulties previously mentioned away. Alas, it also keeps most
practical applications away, for it makes much more practical sense that the value of ${\cal Y}(t)$
is known {\bf before} the application of the input - that is, 
what is given is ${\cal Y}(0^-)$ and not ${\cal Y}(0^+)$. 
To be able to apply this solution we need somehow to determine 
the first condition ${\cal Y}(0^+)$ from the information provided, that is, we
need a map from ${\cal Y}(0^-)$ to ${\cal Y}(0^+)$.

This mapping from $t=0^-$ to $t=0^+$ - that is, from {\em previous} conditions to {\em first} conditions - 
has been amply considered before \cite{Dervisoglu, inconsistentIC} and different solutions/mappings
exist to accommodate different assumptions on the system being modeled. 
In our present setting, we can determine the first conditions from the previous conditions directly from (\ref{Markov}).
Since this equation is valid for any time instant, we can write it for both limits around $t=0$:
\begin{eqnarray}
{\cal Y}(0^-) &=& O \vecx (0^-) + M {\cal U}(0^-) \\
{\cal Y}(0^+) &=& O \vecx (0^+) + M {\cal U}(0^+) 
\end{eqnarray}
But the state is guaranteed to be a continuous function of time, for its derivative is, by definition, always finite. Thus we can write
\beq\label{MarkovZero+}
{\cal Y}(0^+) - M {\cal U}(0^+) = {\cal Y}(0^-) - M {\cal U}(0^-)
\eeq
or, equivalently, 
\beq\label{MarkovZero++}
{\cal Y}(0^+) = {\cal Y}(0^-)  +  M [{\cal U}(0^+) -  {\cal U}(0^-) ] .
\eeq

Equation (\ref{MarkovZero++}) is a simple formula that allows to determine the first conditions ${\cal Y}(0^+)$ from the 
previous conditions ${\cal Y}(0^-)$. But a fundamental observation must be made here:
the previous value of the input ${\cal U}(0^-)$ must also be provided, along with 
the input values for positive times; without this information it is not possible to solve
uniquely the ODE. 
A second observation is also worth making: the first condition $ {\cal Y}(0^+)$ is composed of two parts, one due
to the previous condition of the output - the term ${\cal Y}(0^-)$ - and another one due
to the input or, more specifically, due to the instantaneous change of the input ${\cal U}(0^+) -  {\cal U}(0^-)$.

Once the first condition has been obtained, one can just use it {\em in lieu} of the initial condition in the formula (\ref{ingenua})
or whatever the solution procedure found in the textbooks that one wants to apply. This is illustrated by an example.

\begin{example}\label{segundo_exemplo}
Consider the ODE:
\begin{equation}\label{example_ordem2}
\ddot{y}(t) + 6\dot{y}(t) + 5 y(t) =  \ddot{u}(t)+ 3\dot{u}(t) + 2u(t) ,
\end{equation}
for which we want to determine the output $y(t)~\forall t>0$ under the following conditions.
At $t=0$ the input switches from a sinusoid $u(t)=\cos(t) $ to a ramp $u(t)=t$. 
Right before the input switches, the output was zero and its derivative was equal to one. 
We thus have the following data: 
\begin{itemize}
\item the previous condition of the output ${\cal Y}(0^-) =[ 1~0]^T$;
\item the previous condition of the input ${\cal U}(0^-) =[ -sin (0) ~cos(0)]^T = [ 0 ~~1]^T$;
\item the input $u(t) = t ~\forall t>0$. 
\end{itemize}

Writing equation (\ref{ingenua}) with the initial values replaced by the first values we have
\begin{equation}\label{plm}
Y(s) = \frac{s^2+ 3s + 2}{s^2+6s+5} \frac{1}{s^2} + \frac{ 1}{s^2+6s+5} \{ (s+6) y(0^+) + \dot{y}(0^+) - (s+3) u(0^+) -  \dot{u}(0^+) .
\end{equation}
The Markov parameters are easily identified as $h_0 = 1$ and $h_1 = -3$, so that, after noting that  ${\cal U}(0^+) =[ 1 ~0]^T$,
equation  (\ref{MarkovZero++}) becomes
$$
{\cal Y}(0^+) = \left[\begin{array}{c} 1 \\ 0 \end{array}\right] +
     \left[\begin{array}{cc} 1 & -3 \\ 0 & 1 \end{array}\right] \{ \left[\begin{array}{c} 1 \\ 0 \end{array}\right] - \left[\begin{array}{c} 0 \\ 1 \end{array}\right] \}
     =  \left[\begin{array}{c} 5 \\ -1 \end{array}\right] .
$$
Using this value for the first condition in (\ref{plm}) yields
\begin{equation}\label{plmm}
Y(s) = \frac{s^2+ 3s + 2}{s^2+6s+5} \frac{1}{s^2} + \frac{ -s-2 }{s^2+6s+5}  .
\end{equation}
\cqfd
\end{example}

On the other hand, it is interesting that it is not even necessary to perform explicitly this calculation of first conditions,
as equation (\ref{MarkovZero++}) also allows us to prove the following result.

\begin{theorem}\label{teo:principal}
Consider an ODE in the form (\ref{EDOLcc}); then equation (\ref{ingenua}) yields the same result whether the initial 
conditions (values at $t=0$) are replaced by the first conditions (at $t=0^+$) or by the previous conditions (at $t=0^-$).

{\bf Proof}

The statement clearly corresponds to the verification of the following equality
\begin{eqnarray}
&&  \sum_{j=0}^{n-1} b_j \sum_{k=1}^{n-j} s^{n-k-j} \stackrel{(k-1)}{u} (0^+) + \sum_{i=0}^{n-1} a_i \sum_{k=1}^{n-i} s^{n-k-i} \stackrel{(k-1)}{y} (0^+) \bydef X(0^+) = \nonumber \\
&&  \sum_{j=0}^{n-1} b_j \sum_{k=1}^{n-j} s^{n-k-j} \stackrel{(k-1)}{u} (0^-) + \sum_{i=0}^{n-1} a_i \sum_{k=1}^{n-i} s^{n-k-i} \stackrel{(k-1)}{y} (0^-) \bydef X(0^-).
\label{igualdade}
\end{eqnarray}

One can change the order of summation  in (\ref{igualdade}) to write 
\begin{equation} \label{ingenua1}
X(0^+)  = \sum_{k=1}^{n} \stackrel{(k-1)}{u} (0^+) \sum_{j=0}^{n-k} b_j  s^{n-k-j}  + 
\sum_{k=1}^{n} \stackrel{(k-1)}{y} (0^+) \sum_{i=0}^{n-k} a_i  s^{n-k-i}  
\end{equation}
and similarly for $X(0^-)$.


Equation (\ref{ingenua1}) can also be written in vector form as
\begin{equation}\label{ingenua_vetor}
X(0^+) =
\mathbf{v}^T_y {\cal Y}(0^+) -  \mathbf{v}^T_u {\cal U}(0^+)
\end{equation}
where 
\begin{equation}\label{ves}
\mathbf{v}_y = \left[\begin{array}{c}
			1 \\
			 s + a_1 \\
			\vdots \\
			\sum_{i=0}^{n-k} a_i  s^{n-k-i} \\
			\vdots \\
			\sum_{i=0}^{n-1} a_i  s^{n-1-i} 
			\end{array}\right]
			\quad
\mathbf{v}_u = \left[\begin{array}{c}
			b_0 \\
			b_0 s + b_1 \\
			\vdots \\
			\sum_{j=0}^{n-k} b_j  s^{n-k-j} \\
			\vdots \\
			\sum_{j=0}^{n-1} b_j  s^{n-1-j} 
			\end{array}\right] .
\end{equation}

Hence, equation (\ref{igualdade}) can be written in vector form as
\begin{equation}\label{igualdadevetor}
\mathbf{v}^T_y {\cal Y}(0^+) -  \mathbf{v}^T_u {\cal U}(0^+) = \mathbf{v}^T_y {\cal Y}(0^-) -  \mathbf{v}^T_u {\cal U}(0^-)
\end{equation}

On the other hand, multiplying equation (\ref{MarkovZero+})  to the left by $\mathbf{v}_y$ yields:
\beq
\mathbf{v}_y^T{\cal Y}(0^+) - \mathbf{v}^T_yM {\cal U}(0^+) = \mathbf{v}^T_y{\cal Y}(0^-) - \mathbf{v}^T_yM {\cal U}(0^-)
\eeq
which will be exactly the same as (\ref{igualdadevetor}) (and thus (\ref{igualdade}), which would prove the result) if
$\mathbf{v}^T_yM=\mathbf{v}^T_u$.
To show that this is the case, we recall a property of the Markov parameters:
\begin{equation}
G(s) = \frac{\sum_{j=0}^{n} b_j s^{n-j}}{s^n +\sum_{i=1}^{n} a_i s^{n-i} } = \sum_{i=0}^{\infty} h_i s^{-1} .
\end{equation}
This expression, after rewritten as $\sum_{j=0}^{n} b_j s^{n-j} = (s^n +\sum_{i=1}^{n} a_i s^{n-i})  \sum_{i=0}^{\infty} h_i s^{-1}$,
 allows to express the coefficients $b_j$ as functions of the Markov parameters $h_i$:

\begin{eqnarray}
b_0 & = & h_0 \\
b_1 & = & h_1 + a_1 h_0 \\
b_2 & = & h_2 + a_1 h_1 + a_2 h_0 \\
\vdots && \\
b_j & = & \sum_{i=0}^{j-1} a_{i} h_{j-i} . \label{mnb} 
\end{eqnarray}

But
\begin{eqnarray*}
\mathbf{v}^T_yM & = & \left[\begin{array}{c}
h_0 \\
h_1 + (s + a_1) h_0 \\
h_2 + ( s + a_1) h_1 + (s^2 + a_1 s + a_2 ) h_0 \\
\vdots \\
\end{array}\right]^T  \\ 
& = &  
\left[\begin{array}{c}
h_0 \\
h_0 s + (h_1  + a_1 h_0) \\
h_0 s^2 + ( h_1 + a_1 h_0) s + (h_2 + a_1 h_1 + a_2 h_0) \\
\vdots \\
\end{array}\right]^T
\end{eqnarray*}
which, after using (\ref{mnb}), results in  $\mathbf{v}^T_yM=\mathbf{v}^T_u$.
\cqfd
\end{theorem}

\begin{example}
Consider again Example \ref{segundo_exemplo}, whose solution (\ref{plmm}) was obtained from equation (\ref{plm}).
Now rewrite (\ref{plm}) with the first conditions replaced by the previous conditions, that is:
\begin{equation}\label{plmn}
Y(s) = \frac{s^2+ 3s + 2}{s^2+6s+5} \frac{1}{s^2} + \frac{ 1}{s^2+6s+5} \{ (s+6) y(0^-) + \dot{y}(0^-) - (s+3) u(0^-) -  \dot{u}(0^-) \}.
\end{equation}

Using the previous conditions ${\cal Y}(0^-) =[ 1~0]^T$ and ${\cal U}(0^-) = [ 0 ~~1]^T$ in (\ref{plmn}), as
prescribed in Theorem \ref{teo:principal}, results in the
same expression (\ref{plmm}).\cqfd
\end{example}

So, the answer to the initial question - what value to use in lieu of $t=0$ when this is
not defined - is that we can use either the previous values or the first values of
${\cal U}(t)$ and  ${\cal Y}(t)$, as they both yield the same result. It is just necessary
to be consistent, that is, one can either use ${\cal U}(0^-)$ and  ${\cal Y}(0^-)$ or
${\cal U}(0^+)$ and  ${\cal Y}(0^+)$ - one can not mix first values with previous values.
If the problem data consist of  the previous values  ${\cal Y}(0^-)$ instead of the first value
${\cal Y}(0^+)$, as is more sensible to expect, then one needs to know the previous
value of the input as well - knowing
only $u(t) ~ \forall t>0$ is not enough, one also needs to know ${\cal U}(0^-)$.
On the other hand, if the problem data include the first values ${\cal Y}(0^+)$, 
only $u(t) ~ \forall t>0$ must be known - the effect of the input change at $t=0$
is already taken into account in the value of  ${\cal Y}(0^+)$.
Alternatively, one can take the previous values of ${\cal Y}(0^-)$ and ${\cal U}(0^-)$, 
determine from them the first values ${\cal Y}(0^+)$ and ${\cal U}(0^+)$, and then use them
in the formula or in any other solution method one may prefer. 
This is a more general procedure, in the sense that it applies also to other solution methods.
Moreover, it also applies  when solving a more general class of problems, those of singular systems - but this
is another story, to be told some other time.

\section{Conclusions}

The treatment given in the literature to the solution of linear ODE's in response to nonzero initial conditions
could use some more precision. Specifically, the smoothness of the signals involved - input and
output - and its effects are not properly analyzed in any textbook we could find.
Many of the theorems and formulas presented are valid only under some
smoothness assumptions that are rarely stated in the literature and rarely
satisfied in real life. 
The (lack of) smoothness also poses a more basic, philosophical, question: 
is there any sense in an equation involving real functions in which some terms
do not exist in this realm?
The classical answers to this question rely on sophisticated mathematical
objects - distributions, singular ``functions'' - which have to be presented to the student
for this purpose only.
These difficulties are absent is state space representations, whose solutions are continuous
as long as the input is finite. The literature abounds with analyses of
the equivalence between the input-output behavior of ODE's and SSR's, but it is hard to
find anything about equivalence regarding their responses to initial conditions. 

We have provided in this paper some formal definitions and some results that we believe are useful to clarify
these issues. First, we have made explicit the assumptions behind the solutions most usually found in the 
literature, which consist of severe constraints on the smoothness of the input. We have analyzed the 
equivalence between state space representations and differential equations. Then, based on this
equivalence, we have  defined the solution of a ``psychedelic'' ODE - one whose RHS goes to infinity - 
as the solution of its equivalent SSR, which removes the need to enter the realm of singular ``functions''.
Based on these concepts we were able to provide a simple answer to the question
that has motivated our work in the first place: what to do to solve an ODE when the said smoothness
conditions are not satisfied. Actually one can use either the first or the previous conditions
of the input and output  {\em in lieu} of the initial conditions, though some care must exercised in doing it, as
explained in the text. 

%

%
%


\bibliographystyle{plain}
\bibliography{CI}

\end{document}